\documentclass[journal]{IEEEtran}
\usepackage{mathrsfs}
\usepackage{amsmath}
\usepackage{amsthm}

\usepackage{paralist}
\usepackage{multirow}
\usepackage{graphicx}
\usepackage{amsfonts}
\usepackage{amsmath,epsfig}
\usepackage{stfloats}
\usepackage{color}
\usepackage{amssymb}
\usepackage{epstopdf}
\usepackage{latexsym}
\usepackage{cite}
\usepackage{algorithm}
\usepackage{algorithmic}

\ifCLASSINFOpdf

\else

\fi

\begin{document}
%
\title{Energy-Efficient Buffer-Aided Relaying Systems with Opportunistic Spectrum Access}

\author{Kunlun Wang, \IEEEmembership{Member,~IEEE},
        Qingqing Wu, \IEEEmembership{Member,~IEEE}, Wen Chen, \IEEEmembership{Senior~Member,~IEEE}, Yang Yang, ~\IEEEmembership{Fellow,~IEEE}, and Derrick Wing Kwan Ng, \IEEEmembership{Senior~Member,~IEEE} 
\thanks{The work of K. Wang is supported by the National Natural Science Foundation of China (NSFC) under grant 61801463. The work of W. Chen is supported by the National Key R\&D Program of China under Grant 2018YFB1801102, the NSFC under grant 61671294, the National Huge Project 2018ZX03001009-002, and the STCSM Science and Technology Innovation Program 17510740700. The work of D. W. K. Ng is supported by funding from the UNSW Digital Grid Futures Institute, UNSW, Sydney, under a cross-disciplinary fund scheme and by the Australian Research Council's Discovery Project (DP190101363).}
\thanks{K. Wang, and Y. Yang are with the School of Information Science and Technology, ShanghaiTech University, Shanghai, China, and also with the Shanghai Institute of Fog Computing Technology~(SHIFT), Shanghai, China; (e-mails: \{wangkl2, yangyang\}@shanghaitech.edu.cn).
Q. Wu is with the Department of Electrical and Computer Engineering, National University of Singapore, Singapore; (e-mail: elewuqq@nus.edu.sg).
W. Chen is with the Shanghai Institute of Advanced Communications and Data Sciences, Department of Electronics Engineering, Shanghai Jiao Tong University, China; (e-mail: wenchen@sjtu.edu.cn).
D. W. K. Ng is with the School of Electrical Engineering and Telecommunications,
The University of New South Wales, Sydney NSW 2052, Australia; (e-mail: w.k.ng@unsw.edu.au).

Part of this work is presented at the 2017 International Conference on Wireless Communications and Signal Processing (WCSP), Nanjing, China, October 11-13, 2017 \cite{2017-Wang-cRelay}. 
}
}


\maketitle

\begin{abstract}
In this paper, an energy-efficient cross-layer design framework is proposed for cooperative relaying
networks, which takes into account the influence of spectrum utilization probability. Specifically, random arrival traffic is considered and an adaptive modulation and coding (AMC) scheme is adopted in the cooperative transmission system to improve the system performance. The average packet dropping rate of the relay-buffer is studied at first. With the packet dropping rate and stationary distribution of the system state, the closed-form expression of the delay is derived. Then the energy efficiency for relay-assisted transmission is investigated, which takes into account the queueing process of the relay and the source. In this context, an energy efficiency optimization problem is formulated to determine the optimum strategy of power and time allocation for the relay-assisted cooperative system. Finally, the energy efficient switching strategy between the relay assisted transmission and the direct transmission is obtained, where packet transmissions have different delay requirements. In addition, energy efficient transmission policy with AMC is obtained. Numerical results demonstrate the effectiveness of the proposed design improving the energy efficiency.

\end{abstract}
\IEEEpeerreviewmaketitle

\section{Introduction}
\IEEEPARstart{D}{uring} the past decades, wireless communication systems have been designed with the main focus on improving the network throughput or spectral efficiency \cite{2015-Ling-capacity}. However, with the rapid realization of internet of things (IoT), 
power consumptions at both base stations~(BSs) and resource-limited devices such as user terminals are no longer sustainable, which is attributed to the fundamental bandwidth-power and delay-power tradeoffs \cite{2016-Zhang-green,2018-yang-5G}. As reported, the network energy consumption should be decreased by a factor of $1,000$ in the upcoming fifth-generation (5G) networks while without compromising the quality of sevice (QoS) \cite{2017-Wu-green}. As a result, energy efficiency, which is measured by bits-per-joule, has emerged as a key figure of merit and become the most widely adopted design metric for green IoT systems. In fact, extensive research efforts have been made to improve the energy efficiency for various kinds of wireless communication applications, such as orthogonal frequency division multiplexing access (OFDMA) \cite{2019-Nam-EE,2015-Wu-EE,2013-Kwan-EE,2012-Kwan-EE}, multiple-input multiple-output (MIMO) \cite{2018-Hien-MMIMO,2013-Xu-EE,2013-Rui-EE,2014-Ge-EE}, millimeter-wave (mmWave) communications~\cite{2016-Gao-EEwave,2017-Niu-EE}, 
cognitive radio~\cite{2013-Mao-cognitive,2013-Fu-Cognitive}, small cell~\cite{2017-Zhang-EE,2016-Wu-scell}, 
and wireless powered communications~\cite{2013-Kwan-power,2016-Wu-power}, etc. However, most of the previous research works mainly focus on energy-efficient designs for physical layer transmission. In the future IoT networks, traffic delay requirements are diversified from the upper layers imposing challenges to physical layer design in resource allocation. Thus far, many studies have been carried out on cross-layer design to improve the energy efficiency of the whole system\cite{2011-Qiao-EE,2016-Kwan-QoS, 2015-Nomikos-EE,2017-Wang-cRelay}, which is expected to be more appealing and practical for the next generation of wireless networks.

On the other hand, the gap between the designing of high data rate and spectrum scarcity motivates the advent of cognitive spectrum access, which has been regarded as a revolutionary technology to tackle such a challenge \cite{2010-Yuan-Cognitive}. When the spectrum is detected to be unoccupied by the primary system, the secondary system could utilize the free bands temporarily to perform its own transmission. Then there exists a utilization probability for each transmission link of the secondary system. Furthermore, the utilization probability may influence the buffer-overflow and the delay performance of the packet for the secondary system, leading to different energy efficient strategies. In a cellular network, each secondary user has a chance to transmit data over the free bands. Due to the severe channel fading or shadowing, the selected secondary user may terminate the transmission from the direct link to save energy. In this case, any other nearby secondary users with better channel conditions can act as a relay to assist the data transmission to improve the energy efficiency of the secondary system. Furthermore, since the secondary users may require different delay-aware applications over the cooperative network, providing heterogenous QoS of delays for energy-efficient cooperative systems, has become critical and necessary.

\subsection{Related Works}
Relay nodes can help source node to transmit packets to destination mode in cooperative systems~\cite{2009-Yang-Relay,2018-Zhu-D2DCooperative}, which has been
extensively studied recently. In addition, relay equipped with buffer can further improve system performance and introduce new degrees of freedom for system design~\cite{2014-Zlatanov-buffer,2016-Nikolaos-buffer,2015-Nikola-buffer,2016-Qiao-buffer}. In particular, various cooperative schemes have been proposed to exploit cooperative diversity for improving the physical performance~\cite{2016-Huang-Delay,2011-Ikki-diversity,2011-Amir-diversity}, 
such as outage probability and signal-to-noise ratio~(SNR).
In practical systems, the data link layer performance metric, e.g., delay, also plays an important role in wireless networks. Since delay can influence the performance of the cooperative systems, it also has a significant impact on protocol and system designs, especially when the system is required to guarantee the performance of delay sensitive services and delay tolerant services simultaneously.

The notion of delay in cooperative wireless networks has been considered in recent literatures. In~\cite{2012-Rong-delay}, the performance metrics defined by stable throughput region and delay are evaluated for a packet transmitted to the destination through either a direct link or relay links. The delay-optimal link selection for a two-hop three-node cooperative network with bursty packet arrivals is studied in\cite{2011-cui-delay}. Multi-hop cases are considered in \cite{2013-xue-delay}. Motivated by real-time applications having stringent delay constraints, optimal resource allocation for delay-limited cooperative communication in time varying wireless networks is recently discussed in~\cite{2012-Lin-QoS}.

Most of the existing works on cooperative communications, e.g., \cite{2012-Rong-delay,2011-cui-delay,2013-xue-delay,2012-Lin-QoS}, are based on throughput optimization under delay constraint. However, all these works do not consider the energy efficient communication.
In~\cite{2015-Nomikos-EE}, 
delay constrained energy-efficient problem in cooperative wireless networks is studied. However, practical adaptive modulation and coding~(AMC) is not considered. In~\cite{2012-Wang-Cooperative}, the authors investigate the AMC scheduling for a cooperative wireless system with multiple relays operating in a modified decode-and-forward protocol. However, adaptive scheduling strategy between the cooperative transmission and direct transmission is not considered. Although the scheduling strategy with direct transmission and relay-assisted transmission for cognitive radio network is studied in~\cite{2013-Rao-Cognitive}, the packet delay and energy efficiency are not considered. Furthermore, due to the multi-path channel fading, the noise, and the co-channel interference from primary system, it is difficult to guarantee the deterministic delay for communication services in cognitive radio networks. Thus, as a compromise scheme, it is a practical and ultimate goal to provide statistical delay guarantee for cooperative cognitive radio networks.

\subsection{Main Contributions}
In this work, we consider the cross-layer energy-efficient design approach for cooperative IoT networks with opportunistic spectrum access, which aims at taking the physical layer transmission and delay-aware service into account, the influence of spectrum utilization probability can also be obtained. Generally speaking, the AMC and cooperative transmission strategy are part of the physical layer decisions, the packet retransmission is controlled by the data link layer, and the delay statistic of the traffic is provided by the application layer. In this research direction, there are some works on cross-layer designs considered in~\cite{2012-Wang-Cooperative,2013-Rao-Cognitive,Wangkl-2015-EE}, where a similar system model is considered. In particular, the cross-layer design for cognitive relay systems is investigated in~\cite{2013-Rao-Cognitive}. However, our work is fundamentally
different from that of~\cite{2012-Wang-Cooperative,2013-Rao-Cognitive,Wangkl-2015-EE} in several significant aspects. First, our work aims to study the impact of delay-aware service on the design of energy-efficient scheduling, while \cite{2013-Rao-Cognitive,2012-Wang-Cooperative} focus more on the bit-error-rate (BER)-based scheduling. Second, our objective is to maximize the energy efficiency of the delay-aware service in cooperative IoT networks with opportunistic spectrum access. Based on such an objective, we derive an expression of the expected delay in a closed-form which captures the impacts of the queueing and
transmission delay in both the source and relay buffers. To the best knowledge of the authors, this important aspect has not been reported in the literature~\cite{2Wangkl-2015-EE}, yet. Besides, we have also derived the packet dropping rate for the buffer of the relay node, which is not considered in \cite{2013-Rao-Cognitive,2012-Wang-Cooperative} either. Third, we investigate the performance of energy efficiency by utilizing the AMC transmission and the time allocation for the direct transmission and relay-assisted transmission. Last but not least, for maximizing the system energy efficiency, we derive an optimized threshold by taking spectrum access probability into account, which determines the selection between direct transmission and relay-assisted transmission. This contribution is novel and different from the previous works.

The challenge of this work is to derive the packet dropping rate at the relay node, which can influence the throughput and energy efficiency of the cooperative system. To overcome the challenge, we need to prove the existence of the stationary distribution of the cooperative system, and get the stationary distribution. Therefore, our contribution in this work can be summarized as follows: We obtain the closed-form expression of expected delay for the cooperative transmission system, which considers both the transmission and queueing delay in both the source and relay buffer. Then we obtain the closed-form expression for the system throughput by taking into account the packet drop caused by both the transmission error and the buffers overflow. Correspondingly, the average energy consumption of the cooperative system is also derived. Considering the opportunistic spectrum access, we reveal the intrinsic relationship between the energy-efficient scheduling strategy and different delay requirements. In particular, we show that the energy efficient scheduling strategy is to adaptively switch between direct transmission and relay-assisted transmission based on different delay requirements. At last, we obtain the energy-efficient transmission policies for the inter-node link transmission and direct link transmission.

\subsection{Organization}
The rest of the paper is organized as follows. Section \uppercase\expandafter{\romannumeral2} introduces the system model and problem formulation, including the network model, the queueing model and the traffic model, the
channel model, and the medium access model. In Section \uppercase\expandafter{\romannumeral3}, we investigate the energy-efficiency for the relay assisted transmission and the direct transmission. In Section \uppercase\expandafter{\romannumeral4}, we analyze the energy-efficient scheduling for the delay-aware service based on the theoretical results obtained in Section \uppercase\expandafter{\romannumeral3}. Section \uppercase\expandafter{\romannumeral5} presents numerical results to demonstrate the performance
of the proposed cross-layer design. Finally, we conclude the paper in Section~\uppercase\expandafter{\romannumeral6}.

\begin{figure}[!t]
\begin{center}
\includegraphics [width=2.3in]{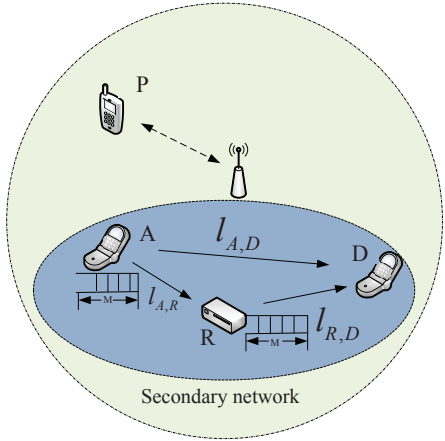}
\caption{System model.} \label{Fig-system_model} 
\end{center}
\end{figure}

\section{System Model}
\subsection{Network Model}
As shown in Fig.~\ref{Fig-system_model}, we consider a cooperative secondary system which consists of three nodes
that share the same spectrum with a legacy primary system (PS) node $P$. Nodes $A$, $R$, and $D$, are the source, relay, and destination of the secondary system, respectively. Specifically, the relay node adopts the decode-and-forward~(DF) strategy to assist the end-to-end transmission. Due to the bursty service in the PS, the secondary system can access the spectrum of the PS in an opportunistic manner with a certain probability. 
We assume that source node $A$ generates data traffic and communicates with node $D$. In the cooperative secondary network, the data packets can be transmitted via the direct link, $l_{A,D}$, or via the relay path composed of links $l_{A,R}$ and $l_{R,D}$ shown in Fig.~\ref{Fig-system_model}. For each link, i.e., $l_{A,D}$, $l_{A,R}$, or $l_{R,D}$, the secondary system will terminate the transmission if one of the following situation occurs: 1) the PS occupies the spectrum during the transmission; 2) the packet error rate~(PER) of the link $l_{A,R}$ exceeds a maximum tolerable level due to deep channel fading.

\subsection{Traffic and Queueing Models}
%

\textcolor{black}{We assume that the incoming traffic of source $A$ is random, which follows Poisson distribution. A wide range of multimedia traffic can be represented with an Poisson model, which is accurate and reasonable~\cite{2014-Mark-MMPP}.} Specifically, the average packet arrival rate from the application layer is $\bar{\lambda}$. Besides, both the source-buffer and the relay-buffer maintain the first-in-first-out (FIFO) queue and the packet streams have the size of $L$ bits.
Assume that the queueing model of the source node or the relay node is a single server M/G/1 queue~\cite{2008-Gross-Queueing}. 

For the source node $A$, the M/G/1 model assumes the single server with Markovian or memoryless arrivals at an average rate $\bar{\lambda}$ and a general service distribution. Due to the general service distribution, denote by
$\chi_n$ the service rate at the physical link corresponding to channel state $n$. In addition, we assume that the set of service state is $\Psi=\left\{\chi_1, \cdots, \chi_N\right\}$. 

For simplicity, we assume that the buffers of the source and the relay have finite capacity in each of which, denoted as $M$. Due to the bursty arrival traffic, the queues of the nodes may be empty sometimes or full some other times. Therefore, for the latter case, the incoming packets may be dropped at both of the source node and the relay node because of the bursty traffic arrived and the finite buffer size. We assume that the packets dropped by the source and relay nodes are regarded as lost.

\subsection{Channel Model}
All inter-node wireless channels are assumed to have Nakagami-$m$ fading with the same fading parameter $m$, and are block fading with additive white Gaussian noise~(AWGN).

In order to get a well defined relation between the channel quality and received signal-to-noise ratio~(SNR), the transmit signal power of the source node is constant at $\bar{e}$. \textcolor{black}{The channel state information (CSI) is supposed to be perfectly known by the receiver with training, then the channel quality can be captured by the received SNR $\mathcal{S}$.} Due to block fading, $\mathcal{S}$ remains invariant within each transmission frame but can vary from frame to frame. It should be noted that the received SNR $\mathcal{S}$ per frame of a fading channel follows a Gamma distribution, which has probability density function~(PDF) as
\begin{equation}\label{pdf}
f_{\bar{\mathcal{S}}}(\mathcal{S})=\frac{m^{m}\mathcal{S}^{m-1}}{\bar{\mathcal{S}}^{m}\Gamma (m)}\mbox{exp}\left (-\frac{m\mathcal{S}}{\bar{\mathcal{S}}} \right ),
\end{equation}
where $\bar{\mathcal{S}}\triangleq E\left \{ \mathcal{S} \right \}$ is the average received SNR and $\Gamma(m)\triangleq \int_{0}^{\infty }t^{m-1}e^{-t}dt$ is the Gamma function~\cite{2004-Liu-AMC}. 
We assume that inter-node channels of $l_{A,R}$ and $l_{R,D}$ are independent identically distributed~(i.i.d.) with the average SNR $\bar{\mathcal{S}}$ and $\bar{\mathcal{S}}_{R}$, respectively, and the average SNR of the direct channel $l_{A,D}$, denoted by $\bar{\mathcal{S}}_{S,D}$.

For packet transmission at each inter-node link, the adaptive modulation and coding~(AMC) is considered. In particular, the total available number of AMC transmission modes is $N$. To select an AMC mode, the received SNR needs to fall into the corresponding interval. Then, we can divide the entire received SNR range into $N+1$ nonoverlapping consecutive intervals, which can be denoted as $\left \{ \mathcal{S}_n \right \}^{N+1}_{n=0}$,
where the boundary points is given by $\mathcal{S}_0=0$ and $\mathcal{S}_{N+1}=+\infty$.
Based on the above description, the channel is in state $n$ satisfying $\mathcal{S}\in[\mathcal{S}_n,\mathcal{S}_{n+1})$. Based on~\eqref{pdf}, we can obtain the probability that the channel being in state $n$ as
\begin{equation}\label{prn}
P_r(n)=\int_{\mathcal{S}_n}^{\mathcal{S}_{n+1} }f_{\bar{\mathcal{S}}}(\mathcal{S})d\mathcal{S}.
\end{equation}
Similarly, the probabilities that the channel being in state $n$ for the relay-to-destination link and the source-to-destination link can be obtained as $P^R_r(n)$ and $P^{S,D}_r(n)$.

Denote by $b_n$ the transmission modulation when the channel is in state $n$.
Specifically, if the receiver cannot correctly decode a packet, the transmitter will be notified to repeat transmitting the packet where the maximum number of retransmissions is $N^{\max}_r-1$.
In contrast, when the receiver correctly decodes a packet, it will feed back an acknowledgement~(ACK) packet to the transmitter. However, in some cases for which the packet cannot be correctly decoded at the receiver even after $N_r^{\max}-1$ times of retransmissions, then this packet will be dropped and a packet loss will be declared. Therefore, packet loss may happen in each link of the cooperative system. \textcolor{black}{With the average packet error rate of transmissions, both the destination node in direct link and the relay node can fail to decode the received packet, an negative acknowledgement~(NACK) packet will be transmitted to notify the source node. If the destination node successfully receives the packet, but the relay node fails to receive the packet, then the ACK packet will be transmitted, the packet will be transmitted through the direct link. On the other hand, if the relay node successfully receives the packet in the relay-assisted transmission, but the direct source-to-destination transmission fails, then the ACK packet will be transmitted, the packet will be transmitted through the relay-assisted link.} In practice, the probability of packet loss is regained to no larger than a threshold $P_{\mathrm{loss}}$ to satisfy the quality-of-service~(QoS). It is worth pointing out that the service rates of the source node and relay node depend not only on the transmission modulation $b_n$ but also the maximum retransmission times $N^{\max}_r-1$.

\subsection{Medium Access Model}
In general, the spectrum occupancy of the PS can be modeled as a continuous-time Markov chain with available~(the link is idle) and unavailable~(the link is busy) states~\cite{2008-Zhao-Opportunistic}. In this work, we assume that the transmissions of the PS and secondary cooperative system are both continuous. Thus, the spectrum occupied times from the secondary cooperative system are independent and exponentially distributed with aggregated parameter $q_{i,j}^{-1}$ for the available state and $u_{i,j}^{-1}$ for the unavailable state on link $l_{i,j}$ ($i,j\in \left\{A,R,D\right\}$). Under this model, the stationary probability of the available and unavailable states are respectively given by~\cite{2008-Zhao-Opportunistic} as
\begin{equation}\label{available}
a_{i,j}^{a}=\frac{q_{i,j}}{q_{i,j}+u_{i,j}},
\end{equation}
and
\begin{equation}\label{unavailable}
a_{i,j}^{u}=\frac{u_{i,j}}{q_{i,j}+u_{i,j}},
\end{equation}
respectively.

\subsection{Problem Statement}
In the considered cooperative network, the system power consumption for transmitting one packet in the secondary network is denoted by $\mathcal{P}_{(T)}$, which indicates the total amount of energy consumed in one second. Accordingly, the system throughput is denoted by $\mathcal{R}$, which indicates the number of packets which are transmitted successfully without error in one second. In general, the energy efficiency is defined as the ratio of average system throughput over the total system power consumption, i.e.,
\begin{equation}\label{eedefinition}
\eta_{ee} \triangleq \frac{\mathcal{R}}{\mathcal{P}_{(T)}}\quad\mbox{packets/joule},
\end{equation}
which indicates the number of successfully received packets per joule of energy consumed.

Since different delay requirements of the traffic and spectrum access probabilities may result in different energy efficiencies of the cooperative network, we should optimize the selection of transmission mode between direct transmission and relay-assisted transmission to maximize the system energy efficiency. While the maximum tolerable delay is $\mathcal{D}_0$, we call the selection of transmission mode switching strategy $\varphi$. For the packet transmission, AMC schemes are used for both the source node $A$ and relay node $R$. Then, we are interested in getting the rate adaption policies maximizing the energy efficiency by taking into account the maximum tolerable packet delay $\mathcal{D}_0$. Specifically, the rate adaption policy is explicitly represented by the probability distribution of modulation $b_n$, $n=\{1,\cdots,N\}$. Denote the rate adaptive policy as $\omega$, which depends on the $\mathcal{S}_n$ and $\bar{\mathcal{S}}$.
Mathematically, the design of transmission mode switch strategy and AMC rate adaption policy can be formulated as
\begin{equation}\label{eeproblem}
\max_{\left\{P_r(n)\right\},\varphi}\left\{\eta_{ee}| \mathcal{D}_{(T)}\leq \mathcal{D}_0, \varphi\in \left\{d,r\right\} \right\},
\end{equation}
where $\mathcal{D}_{(T)}$, $d$ and $r$ represent the average packet delay from the arrival in the source buffer until the destination correctly received, the direct transmission and the relay assisted transmission, respectively. Based on~\eqref{prn}, the problem is equivalent to $\max_{\bar{\gamma},\varphi}\left\{\eta_{ee}| \mathcal{D}_{(T)}\leq \mathcal{D}_0, \varphi\in \left\{d,r\right\} \right\}$.

\section{Energy Efficiency Analysis}
In this section, we evaluate the delay and the energy efficiency of a packet transmission in the cooperative network with opportunistic spectrum access.
When the direct transmission is not possible, i.e., the received SNR is smaller than a threshold, the source node selects the relay to assist the transmission to achieve better performance. We denote the total time duration of each transmission period as $T\triangleq T_1+T_2$ and denote the ratio of the time allocation in the phase of source-to-relay over the whole period as $\alpha\triangleq T_1/T\in (0,1)$, where $T_1$ and $T_2$ are the transmission time durations for source-to-relay link and relay-to-destination link, respectively.

\subsection{Throughput Analysis}
\subsubsection{Service Rate}
The source node selects transmission modulation based on AMC, which is also adopted by the relay node. Intuitively, due to the long distance between the source and destination or limited power at the source, the packets transmission is assisted by a relay. At the destination, the receiver attempts to decode the packet from the source or the relay. Regarding the successfully receiving packets, an ACK packet will be transmitted to notify the source node at the end of the time slot. Once the source node receives the ACK packet, the corresponding packet is removed from the source buffer.

\textcolor{black}{Assume each packet contains $L$ bits. The transmission packet rate $R_p$ can be obtained as
\begin{equation}
R_p=\frac{R_sb^{l_{A,D}}_n}{L}.
\end{equation}
Now, with the packet error rate $PER_n$ at channel state $n$ and the probability that the channel being in state $n$, the average number of error transmitted packets can be derived as
\begin{equation}\label{error}
\sum_{n=1}^{N}\frac{R_sb^{l_{A,D}}_n}{L}P^{S,D}_r(n)PER_n.
\end{equation}
Similarly, the average total number of transmitted packets can be obtained as
\begin{equation}\label{total}
\sum_{n=1}^{N}\frac{R_sb^{l_{A,D}}_n}{L}P^{S,D}_r(n).
\end{equation}
With \eqref{error} and \eqref{total}, we can derive the average packet error rate of the direct source to destination link as
\begin{equation}
P_{LD}=\frac{\sum_{n=1}^{N}\frac{R_sb^{l_{A,D}}_n}{L}P^{S,D}_r(n)PER_n}{\sum_{n=1}^{N}\frac{R_sb^{l_{A,D}}_n}{L}P^{S,D}_r(n)}.
\end{equation}
Similarly, the average packet error rates of the source-to-relay link and the relay-to-destination link via the help of the relay can be derived as
\begin{equation}\small
\begin{aligned}
P_{L1}=\frac{\sum_{n=1}^{N}\frac{R_sb^{l_{A,R}}_n}{L}P^{S,D}_r(n)PER_n}{\sum_{n=1}^{N}\frac{R_sb^{l_{A,R}}_n}{L}P^{S,D}_r(n)},\\
P_{L2}=\frac{\sum_{n=1}^{N}\frac{R_sb^{l_{R,D}}_n}{L}P^{R}_r(n)PER_n}{\sum_{n=1}^{N}\frac{R_sb^{l_{R,D}}_n}{L}P^{R}_r(n)},
\end{aligned}
\end{equation}
respectively.}
Note, the metric of average packet error rate is defined by the ratio of the average number of packets cannot be successfully received over the total average number of transmitted packets, and the packet error rate $PER_n$ at channel state $n$ is denoted as~\cite[Eq. (5)]{2012-Wang-Cooperative}
\begin{equation}
PER_n=\left\{\begin{matrix}
\alpha_n \mbox{exp}(-g_n\mathcal{S}), & \mathcal{S}\geq \mathcal{S}_{pn}\\
1, & \mathcal{S}<\mathcal{S}_{pn}
\end{matrix}\right.
\end{equation}
and $\alpha_n$, $g_n$, and $\mathcal{S}_{pn}$ are mode dependent parameters.

Considering the single packet transmission with two phases, the destination node and the relay node may receive this packet. With the average packet error rate of transmissions, the probability that both the destination node and the relay node fail to decode the received packet in the first phase is given by
\textcolor{black}{\begin{equation}\label{error1}
P_1=P_{LD} P_{L1}.
\end{equation}}
On the other hand, if we assume that the relay node successfully receives the packet in the first phase, but the direct source-to-destination transmission fails in the first phase, then the packet cannot be successfully received at the destination with probability
\begin{equation}\label{error2}
P_2=P_{LD}(1-P_{L1})P_{L2}.
\end{equation}
Based on~\eqref{error1} and \eqref{error2}, the overall probability that a single packet cannot be successfully received by the destination via the relay-assisted transmission can be obtained as follows
\begin{equation}
P_0=P_1+P_2=P_{LD}\cdot P_L+P_{LD}(1-P_{L1})P_{L2}.
\end{equation}

Now, we study the packet's service rate via the following two lemmas.
\newtheorem{lemma}{Lemma}
\begin{lemma}\label{total-service}
For the relay-assisted transmission at channel state $n$, the packet's service rate from the source is
\begin{equation}\small
\chi_n=\frac{\alpha(1-P_1)}{f(\epsilon_n)+\bar{\epsilon}_nP_1\left [ 1-P_1^{N^{\max}_r}\left(1+N^{\max}_r (1-P_1)\right) \right]},
\end{equation}
where $f(\epsilon_n)=\epsilon_n(1-P^{N^{\max}_r+1}_1)(1-P_1)$ and $\epsilon_n$ is shown in \eqref{r-tau} at top of the page, which is the transmission time of the packet.
\begin{figure*}[!t]
\normalsize
\begin{equation}\small\label{r-tau}
\begin{aligned}
\epsilon_n&=\tau_n^{l_{A,D}}(1-P_{LD})+\tau_n^{l_{A,R},l_{R,D}}P_{LD}\\
&=\frac{L}{b_nR_{s}a^a_{A,D}}(1-P_{LD})+P_{LD}\left(\frac{L}{b_nR_{s}a^a_{A,R}}+\frac{L}{\bar{b}R_{s}a^a_{R,D}}\right)
\end{aligned}
\end{equation}
\hrulefill
\vspace*{10pt}
\end{figure*}
\end{lemma}
\begin{proof}[Proof]
The proof is given in Appendix~\ref{a1}.
\end{proof}
\begin{lemma}\label{relay-service}
For the relay-assisted transmission, the packet's service rate from the relay to the destination at channel state $n$ is \begin{equation}\small
\chi^{'}_n=\frac{(1-\alpha)(1-P_{L2})}{f(\epsilon^{'}_n)+\bar{\epsilon}^{'}_nP_{L2}\left [ 1-P_{L2}^{N^{\max}_r}\left(1+N^{\max}_r (1-P_{L2})\right) \right]},
\end{equation}
where $f(\epsilon^{'}_n)=\epsilon^{'}_n(1-P^{N^{\max}_r+1}_{L2})(1-P_{L2})$,  $\epsilon^{'}_n=\frac{L}{b^{l_{R,D}}_nR_{s}a^a_{R,D}}$ and $\bar{\epsilon}^{'}_n=\frac{L}{\bar{b}^{l_{R,D}}_nR_{s}a^a_{R,D}}$ represent the relay to the destination packet transmission time and average retransmission time for each retransmission, respectively.
\end{lemma}
\textcolor{black}{\begin{proof}
Similar to the proof of Lemma~\ref{total-service}.
\end{proof}}
Finally, for the source and the relay node, we obtain the sets of service rates as $\Psi=\left\{\chi_1, \cdots, \chi_N\right\}$ and $\Psi^{'}=\left\{\chi^{'}_1, \cdots, \chi^{'}_N\right\}$, respectively.

\subsubsection{Packet Dropping Rate}
Due to the bursty service, the packet can be dropped from the buffer due to buffer overflow. Now, the packet dropping rate at the source node $A$ and the relay node $R$ can be written as $P^A_d$ and $P^R_d$, respectively. Intuitively, when the number of packet arrivals is larger than the remaining space of the buffer, some packets will be dropped. At time $t$, assumed that the current service rate is $\chi^{'}_t$ at relay node $R$, and the service rate is $\chi^{'}_{t-1}$ at previous time. Then, the remaining space in the buffer of node $R$ is sized by $r^R_t=M-(\chi^{'}_{t-1}-\chi^{'}_t)$. As a result, the buffer for the relay node can accommodate at most $r^R_t$ arriving packets in the current time slot. Now, under the condition that the number of arriving packets $a^R_t$ for relay node is larger than $r^R_t$, there are $a^R_t-r^R_t$ packets need to be dropped.


\newtheorem{theorem}{Theorem}
\begin{theorem}\label{t1}
Under the condition that the average transmission rate from the source to the relay $\bar{\chi}>M-(S^R_{t-1}-\chi^{'}_t)$, the packet dropping rate at relay node $R$ is $P^R_d=\frac{E\left\{\Delta^R\right\}}{\bar{\chi}}$, where
\begin{equation}\small\nonumber
\begin{aligned}
E\left\{\Delta^R\right\}&=\sum_{s^{'}\in \mathbb{S}, \chi^{'}\in \Psi^{'}}\Delta^R\cdot\ss(\bar{\chi}|S^R_{t-1}=s^{'}, \chi^{'}_t=\chi^{'})\\
&=\sum_{ s^{'}\in \mathbb{S}, \chi^{'}\in \Psi^{'}}\left [(\bar{\chi}-(M-(S^R_{t-1}-\chi^{'}_t)))\times \pi^{'}_{s^{'},\chi^{'}}\right ]
\end{aligned}
\end{equation}
is the average number of dropped packets at relay node $R$ and $\pi_{s^{'},\chi^{'}}$ is the stationary distribution of the buffer state and the service rate state for the relay to the destination transmission system.
\end{theorem}
\begin{proof}[Proof]
The proof is given in Appendix~\ref{a2}.
\end{proof}

In order to derive the packet dropping rate $P^A_d$ at the source buffer, the stationary distribution $\pi$ of the buffer system should also be considered.
At the time $t$, the number of arrival packets at source is $a^A_t$. Considering that the number of packets dropped at source node $A$ is $\Delta^A_t$, the packet dropped rate $P^A_d$ can be computed as~\cite{2005-Liu-AMC}
\begin{equation}\label{droprate}
P^A_d\triangleq \lim_{T\rightarrow \infty}\frac{\sum_{t=1}^{T}\Delta^A_t}{\sum_{t=1}^{T}a^A_t}=\frac{E\left\{\Delta^A\right\}}{E\left\{a^A_t\right\}}=\frac{E\left\{\Delta^A\right\}}{\bar{\lambda}}.
\end{equation}
Based on Theorem~\ref{t1}, the average number of dropped packets $E\left \{ \Delta^A \right \}$ can be derived as in~\eqref{dropnum-relay}.
\begin{equation}\label{dropnum-relay}
\begin{aligned}
E\left \{ \Delta^A \right \}&=\sum_{ s\in \mathbb{S}, \chi\in \Psi}\Delta^A\cdot\ss(a_t=a, S^A_{t-1}=s, \chi_t=\chi)\\
&=\sum_{ s\in \mathbb{S}, \chi\in \Psi}\left [(a^A_t-(M-(S^A_{t-1}-\chi_t)))\times \pi_{s,\chi}\right ]
\end{aligned}
\end{equation}

With $P^A_d$ and $P^R_d$ available, we can get the average traffic rate and the system throughput. \textcolor{black}{By considering the influences of both the source-queueing and the relay-queueing, we will analyse the packet delay in the next subsection.}

\subsubsection{Network Throughput}
To analyze the performance of the network throughput, we need to consider how the packets are successfully transmitted to the destination from the cooperative network. During the packet transmission, the packet dropping rate from the source and relay queueing and packet violation from the channel with $N^{\max}_r$ retransmissions are influencing the successfully transmission. This means that the system average throughput is related to not only the packet dropping rate $P^A_d$ and $P^R_d$, but also the average overall packet error rate $P_0$. Then, for an average packet arrival rate of $\bar{\lambda}$, we can obtain the system average throughput $\mathcal{R}$ as
\begin{equation}\label{at}
\mathcal{R}=\bar{\lambda}(1-P^A_d)(1-P^R_d)(1-P_0^{N^{\max}_r}).
\end{equation}

\subsection{Delay Analysis}
After studying the network throughput in the previous subsection, we now focus on the average packet delay for the relay-assisted network in details, which considers the spectrum access probability. 

By considering the stationary distribution of the source-buffer system, the average queue length of source node $A$ at channel state $n$ can be obtained as
\begin{equation}\label{average-queue}
\bar{Q}^n_{(q,A)}=\sum_{s\in \mathbb{S}}\sum_{\chi=\chi_n}\pi_{(s,\chi)}\cdot s.
\end{equation}
At the same time, by considering the stationary distribution $\pi_{(s^{'},\chi^{'})}$ of the relay-buffer system, we can get the average queue length of the relay node $R$ at channel state $n$ as
\begin{equation}\label{average-queue}
\bar{Q}^n_{(q,R)}=\sum_{s^{'}\in \mathbb{S},\chi^{'}=\chi^{'}_n}\pi_{(s^{'},\chi^{'})}\cdot s^{'}.
\end{equation}

With regard to the M/G/1 queue, it is known that the average waiting time of a packet consists of queueing time and service time. According to Little's formula~\cite{2008-Gross-Queueing}, the queueing delay is $\bar{\mathcal{D}}_{q}=\frac{\bar{Q}_{q}}{r}$.
Then, we can get the average packet delay $\bar{\mathcal{D}}_n$ for channel state $n$ given by~\eqref{delay-relay}, where $\bar{\chi}=\sum_{n=1}^{N} P_r(n)\chi_n$.
\begin{figure*}[!t]
\normalsize
\begin{equation}\small\label{delay-relay}
\bar{\mathcal{D}}_n=\frac{\bar{Q}^n_{(q,A)}}{\bar{\lambda}}+\frac{\bar{Q}^n_{(q,R)}}{\bar{\chi}}+E\left\{\delta_n\right\}=\frac{\bar{Q}^n_{(q,A)}}{r}+\frac{\bar{Q}^n_{(q,R)}}{\bar{\chi}}+\frac{f(\epsilon_n)+\bar{\epsilon}_nP_1\left [ 1-P_1^{N^{\max}_r}\left(1+N^{\max}_r (1-P_1)\right) \right]}{1-P_1}.
\end{equation}
\hrulefill
\vspace*{10pt}
\end{figure*}
In order to simplify the calculation, the same channel state can be considered for the inter-node channel. Then, we can obtain the average packet delay over the relay-assisted transmission system as
\begin{equation}\label{adelay}
\mathcal{D}_{(T)}=\sum_{n=1}^{N}P_r(n)\bar{\mathcal{D}}_n.
\end{equation}
\begin{lemma}\label{delay-snr}
For the cooperative transmission systems, the average delay $\mathcal{D}_{(T)}$ is a decreasing function of the average received SNR $\bar{\mathcal{S}}$.
\end{lemma}
\textcolor{black}{\begin{proof}
Lemma~\ref{delay-snr} is intuitively correct. Notice that when increasing $\bar{\mathcal{S}}$, the probabilities of choosing larger transmission rate for both the source node and the relay node are increasing. According to \eqref{delay-relay}, $\bar{\mathcal{D}}_n$ is decreasing for $\forall n$. As a result, $\mathcal{D}_{(T)}$ is decreasing. Particularly, the delay can be formulated as $\mathcal{D}_{(T)}=\phi (\bar{\mathcal{S}})$, and $\phi (\bar{\mathcal{S}})$ is a decreasing function of $\bar{\mathcal{S}}$. 
\end{proof}}

\subsection{Power Consumption}
Let $\mathcal{P}_n$ denote the transmission power for the source node at channel state $n$, and the BER $\vartheta$ of the packet can be expressed as a function of the post-adaptation received SNR $\mathcal{S}\mathcal{P}_n/\bar{e}$~\cite{2009-Jalil-AMC}
\begin{equation}\label{ber}
\vartheta\approx 0.2\mbox{exp}(-\frac{1.5}{2^{b_n}-1}\frac{\mathcal{P}_n}{\bar{e}}\mathcal{S}),
\end{equation}
where $b_n$ is the modulation size for AMC mode $n$.

Using~\eqref{ber}, we obtain the transmission power $\mathcal{P}_n$ as a function of the BER $\vartheta$ and $\mathcal{S}$ as follows:
\begin{equation}\label{avp}
\mathcal{P}_n=\frac{\bar{e}(2^{b_n}-1)}{1.5 \mathcal{S}}\mbox{ln}\frac{0.2}{\vartheta}.
\end{equation}
Thus, the average transmission powers of source node and relay node at channel state $n$ can be respectively obtained as
\begin{equation}\label{modep}
\begin{aligned}
\bar{\mathcal{P}}_n=\int_{\mathcal{S}_n}^{\mathcal{S}_{n+1} }\frac{\bar{e}(2^{b_n}-1)}{1.5 \mathcal{S}}\mbox{ln}\frac{0.2}{\vartheta}f_{\bar{\mathcal{S}}}(\mathcal{S})d\mathcal{S},\\
\text{and}\quad \bar{\mathcal{P}}^R_n=\int_{\mathcal{S}_n}^{\mathcal{S}_{n+1} }\frac{\bar{e}(2^{b_n}-1)}{1.5 \mathcal{S}}\mbox{ln}\frac{0.2}{\vartheta}f_{\bar{\mathcal{S}}_R}(\mathcal{S})d\mathcal{S},
\end{aligned}
\end{equation}
where BER can be obtained by $PER_n$.
Based on \eqref{prn} and \eqref{modep}, the average transmit power $\mathcal{P}$ of source node can be shown as
\begin{equation}
\mathcal{P}=\sum_{n=1}^{N}\bar{\mathcal{P}}_nP_r(n),
\end{equation}
Similarly, we can obtain the average transmit power $\mathcal{P}_R$.

\subsection{Energy Efficiency of Relay-Assisted Transmission}
Note that the transmit power for the relay node is $\mathcal{P}_R$. Furthermore, we assume that the power consumption for the node at idle state is constant $\mathcal{P}_0$. Then, the total energy consumption can be obtained by the following lemma.

\begin{lemma}\label{total-power}
For the secondary relay-assisted transmission networks, the total energy consumption with opportunistic spectrum access for each transmission period is
\begin{equation}\label{power}
\begin{aligned}
\zeta=(a^a_{A,R}\mathcal{P}+a^u_{A,R}\mathcal{P}_0)T_1+(a^a_{R,D}\mathcal{P}_R+a^u_{R,D}\mathcal{P}_0)T_2.
\end{aligned}
\end{equation}
\end{lemma}
\textcolor{black}{\begin{proof}
The proof is given in Appendix~\ref{a3}.
\end{proof}}

Based on \eqref{at} and lemma~\ref{total-power}, the energy efficiency is given by
\begin{multline}\small\label{energy-efficiency-relay}
\eta_{ee}=\frac{T\bar{\lambda}(1-P^A_d)(1-P^R_d)(1-P_0^{N^{\max}_r})}{\zeta}=\\
\frac{T\bar{\lambda}(1-P^A_d)(1-P^R_d)(1-P_0^{N^{\max}_r})}{(a^a_{A,R}\mathcal{P}+a^u_{A,R}\mathcal{P}_0)T_1+(a^a_{R,D}\mathcal{P}_R+a^u_{R,D}\mathcal{P}_0)T_2}.
\end{multline}
We aim to jointly optimize the average SNR $\bar{\mathcal{S}}$, $\bar{\mathcal{S}}_R$ and the time allocations for the cooperative transmission based on the closed-form expression of the energy efficiency. Note that the time allocation ratio $\alpha$ may take any number in the range of $(0,1)$. The optimization problem can be specified as follows:
\begin{equation}\label{eeproblem}
\begin{aligned}
&\max_{\alpha, \bar{\mathcal{S}},\bar{\mathcal{S}}_R}\quad \eta_{ee}\\
&\mbox{s.t.} \begin{array}[t]{rcl}
    0 &<&\alpha\quad< 1,\\
    \bar{\mathcal{S}}_{\min} &<&\bar{\mathcal{S}},\bar{\mathcal{S}}_R< \bar{\mathcal{S}}_{\max}.
       \end{array}
\end{aligned}
\end{equation}

For any given time allocation ratio $\alpha\in (0,1)$, we found that we are able to express the corresponding energy efficient transmit power $\mathcal{P}$ and $\mathcal{P}_R$ in terms of the time allocation ratio $\alpha$ with closed-form expressions, which are respectively denoted as $\mathcal{P}^{\ast}(\alpha)$ and $\mathcal{P}_R^{\ast}(\alpha)$. Since the variables $\alpha$, $\bar{\mathcal{S}}$ and $\bar{\mathcal{S}}_R$ are coupled to the packet dropping rate and the transmit power, it's also too complex to solve \eqref{eeproblem}. Then, we provide a suboptimal solution for problem~\eqref{eeproblem} in two stages. In the first stage, we fix $\alpha$ for some value in the interval $(0,1)$. To this end, a numerical search algorithm based on a Golden Section search method~\cite{2010-Eligius-Nonlinear} can be utilized to get the energy efficient solution $\bar{\mathcal{S}}^{opt}$ and $\bar{\mathcal{S}}_R^{opt}$. In the second stage, we still apply numerical search of the single variable $\alpha$ over the interval $(0,1)$ to obtain the energy efficient time allocation ratio $\alpha^{\ast}$ that maximizes the energy efficiency $\eta_{ee}$. To summarize, we illustrate the numerical search based procedure in the following Algorithm~\ref{alg:framwork}.
\begin{algorithm}\small
\caption{Energy efficient time allocation ratio and power allocation for the proposed relay assisted transmission.}
\label{alg:framwork}     
\begin{algorithmic}[1]
\STATE For $\alpha$=0:1\\
$\textbf{Step 1)}$ Initialize the energy efficient threshold value $\mathcal{S}_n=\mathcal{S}^{\ast}_n$~(see \cite{Wangkl-2015-EE}), $\forall n=1,2,\cdots,N$.\\
$\textbf{Step 2)}$ Repeat:\\
calculate $\eta_{ee}$ using \eqref{energy-efficiency-relay};\\
update $\bar{\mathcal{S}}$ and $\bar{\mathcal{S}}_R$ using the Golden Section search method;\\
$\textbf{Step 3)}$ Until $\eta_{ee}$ converge.
\STATE End.
\STATE Energy efficient time allocation ratio:
$\alpha^{\ast}$=$\arg\max \eta_{ee}$, and power allocation for the source node and relay node: $\mathcal{P}^{\ast}$=$\left\{\mathcal{P}\right\}_{\bar{\mathcal{S}}=\bar{\mathcal{S}}^{\mathrm{opt}},\bar{\mathcal{S}}_R=\bar{\mathcal{S}}_R^{\mathrm{opt}}}$.
\end{algorithmic}
\end{algorithm}

\subsection{Energy Efficiency of Direct Link Transmission}
Based on \eqref{avp}, the average transmit power consumption of direct link at channel state $n$ is calculated as
\begin{equation}\small
\bar{\mathcal{P}}_n^{S,D}=\int_{\mathcal{S}_n}^{\mathcal{S}_{n+1} }\frac{\bar{e}(2^{b_{S,D}}-1)}{1.5 \mathcal{S}_{S,D}}\mbox{ln}\frac{0.2}{\vartheta_{S,D}}f_{\bar{\mathcal{S}}_{S,D}}(\mathcal{S})d\mathcal{S},
\end{equation}
where $b_{S,D}$, $\bar{\mathcal{S}}_{S,D}$ and $\vartheta_{S,D}$ denote the modulation size, the average SNR and the BER for the direct source-to-destination transmission, respectively.
Then, the average transmit power for direct link transmission can be obtained as
\begin{equation}
\mathcal{P}_{S,D}=\sum_{n=1}^{N}\bar{\mathcal{P}}_n^{S,D}P^{S,D}_r(n).\label{dpower}
\end{equation}
According to \eqref{dpower} and lemma~\ref{total-power}, the total energy consumption with direct transmission period $T$ can be written as
\begin{equation}\small\label{power-direct}
\begin{aligned}
\zeta=(a^a_{A,D}\mathcal{P}_{S,D}+a^u_{A,D}\mathcal{P}_0)T
\end{aligned}
\end{equation}
where $\mathcal{P}_0$ is power consumption for the source-node being at idle state.

Based on the packet dropping rate $P_d$ of the source node and \eqref{at}, we can get the throughput of the direct transmission system as
\begin{equation}\label{at-direct}
\mathcal{R}=\bar{\lambda}(1-P_d)(1-(1-\bar{P}_{LD})^{N^{\max}_r}).
\end{equation}

Based on \eqref{at-direct} and~\eqref{power-direct}, the energy efficiency can be denoted as
\begin{equation}\label{energy-efficiency}
\begin{aligned}
\eta_{ee}&=\frac{T\bar{\lambda}(1-P_d)(1-(1-\bar{P}_{LD})^{N^{\max}_r})}{\zeta}\\
&=\frac{\bar{\lambda}(1-P_d)(1-(1-\bar{P}_{LD})^{N^{\max}_r})}{a^a_{A,D}\mathcal{P}_{S,D}+a^u_{A,D}\mathcal{P}_0}.
\end{aligned}
\end{equation}
Similarly, we aim to optimize the average SNR $\bar{\mathcal{S}}_{S,D}$ to achieve the energy-efficient direct link transmission based on the closed-form expression of the energy efficiency, which can also be solved with the Algorithm~\ref{alg:framwork}.

\section{Energy Efficient Scheduling}
In this section, we investigate the relationship between the transmission mode selection and packet delay demand, where the transmission mode includes the direct link transmission and relay-assisted cooperative transmission. Then we determine the energy efficient modulation scheduling strategy with the effects of the delay requirements. 

For scheduling heterogeneous flows, the evaluation of end-to-end delays involves in avionics, multimedia~(video and audio) and best-effort data virtual links~\cite{2012-Hua-Delay}. Their performance depends upon scheduling policy, and different scheduling policies between direct transmission and relay-assisted transmission have different results of energy efficiencies. Therefore, we should determine the energy efficient scheduling policy for different flows having different delay requirements.

\begin{theorem}\label{t2}
In the opportunistic spectrum access cooperative network, to ensure energy efficient packet transmission, there exists a delay threshold $\mathcal{D}^{\star}=\phi(\bar{\mathcal{S}}^{\star})$, which determines the selection between two transmission modes. If the access probabilities satisfy $a^a_{A,D}>\frac{a^a_{A,R}+a^a_{R,D}\cdot\alpha}{1+\alpha}$, while the delay requirement satisfies $\mathcal{D}_0<\mathcal{D}^{\star}$, the packet is transmitted only over direct link $l_{A,D}$; While the delay requirement satisfies $\mathcal{D}_0>\mathcal{D}^{\star}$, the packet is transmitted over the cooperative links $l_{A,R}$ and $l_{R,D}$. On the other hand, if $a^a_{A,D}<\frac{a^a_{A,R}+a^a_{R,D}\cdot\alpha}{1+\alpha}$, while the delay requirement satisfies $\mathcal{D}_0<\mathcal{D}^{\star}$, the packet is transmitted over the cooperative links $l_{A,R}$ and $l_{R,D}$; While the delay requirement satisfies $\mathcal{D}_0>\mathcal{D}^{\star}$, the packet is transmitted only over direct link $l_{A,D}$.
\end{theorem}
\begin{proof}
The proof is given in Appendix~\ref{a4}.
\end{proof}

\newtheorem{proposition}{Proposition}
\begin{proposition}(Adopted from Theorem~\ref{t2})\label{l5}
The energy efficient modulation scheduling policy is $\omega^{opt}=\omega(\mathcal{S}^{\ast}_n,\bar{\mathcal{S}}^{opt}(\mathcal{D}_0))$ when $\mathcal{D}_0>\mathcal{D}^{\star}$. On the other hand, the energy efficient modulation scheduling policies are $\omega^{opt}=\omega(\mathcal{S}^{\ast}_n,\bar{\mathcal{S}}^{opt}(\mathcal{D}_0))$ and $\omega^{opt}=\omega(\mathcal{S}^{\ast}_n,\bar{\mathcal{S}}^{opt}_R(\mathcal{D}_0))$ for the source node and relay node when $\mathcal{D}_0<\mathcal{D}^{\star}$.
\end{proposition}

\textcolor{black}{Proposition~\ref{l5} is obtained by the direct transmission and relay-assisted transmission solutions, respectively, where $\bar{\mathcal{S}}^{opt}(\mathcal{D}_0)$ is the energy-efficient average SNR when the maximum tolerable delay is $\mathcal{D}_0$.}
Based on Theorem~\ref{t2} and Proposition~\ref{l5}, we can get the energy efficient solution of switching and modulation scheduling policies for different delay-aware services.



\section{Simulation Results}
In this section, we validate the accuracy of the analytical results for delay-aware energy efficient scheduling over the cooperative wireless system by means of simulations. In all the simulations, the opportunistic spectrum access is considered.

\subsection{System Parameters}
The corresponding simulation parameters are shown as follows: the number of traffic states is $K=2$, the channel states for the inter-node link and direct link is $N=7$, and the buffer size is $M+1=51$ for the source node and relay node; The packet size $L=100$, the maximum retransmission times of packet is $N^{\max}_{r}=6$.
The maximum number of packet arrivals for source node is set as $A=15$, with an average arrival rate, $\lambda_1=1$ packets/time-unit and $\lambda_2=2$ packets/time-unit, the symbol rate $R_s=100$kHz; 
The Nakagami parameter $m=1$, which corresponds to Rayleigh fading channel with no line of sight~(LOS) component. For mathematical tractability, we assume $\bar{\mathcal{S}}_R=\bar{\mathcal{S}}$. 
%
\subsection{Performance Evaluation}
In Fig.~\ref{ee_ratio}, the metric of energy efficiency based on~\eqref{energy-efficiency-relay} is investigated over different values of time allocation ratio $\alpha$. For any time allocation ratio $\alpha\in (0,1)$, we can get the energy efficient SNR as $\bar{\mathcal{S}}^{\mathrm{opt}}(\alpha)$ by Algorithm~\ref{alg:framwork}, and the corresponding energy efficient power allocation can be calculated based on~\eqref{avp} and \eqref{dpower}. In this figure, we plot the energy efficiency based on energy efficient partition method~(EEP) as well as the minimum SNR required to acheive $P_{\mathrm{target}}$~(MSRE) method for comparison~(see \cite{Wangkl-2015-EE}). We can observe from this figure that the energy efficiency of EEP is better than that of MSRE, since we use the energy efficient threshold values for choosing the modulation. In the case of average SNR $\bar{\mathcal{S}}=5$ dB, Fig.~\ref{ee_ratio} shows that the optimal time allocation ratio is $\alpha^{\ast}=0.53$. In the case of $\bar{\mathcal{S}}=8$ dB, the optimal time allocation ratio is the same, since we assume that the average SNR for the source to the relay is the same with that of the relay to the destination.
\begin{figure}[!t]
\begin{center}
\includegraphics [width=2.5in]{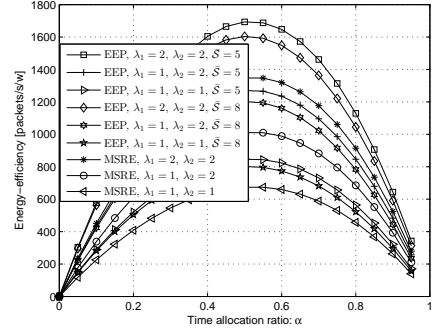}
\caption{The average energy efficiency versus the time allocation ratio $\alpha$ for relay-assisted transmission under the energy efficient partition method~(EEP) and minimum SNR required to acheive $P_{\mathrm{target}}$~(MSRE) method, where $P_{\mathrm{target}}=0.001^{\frac{1}{N_r^{\max}+1}}$.} \label{ee_ratio}
\end{center}
\end{figure}

\textcolor{black}{Denote the spectrum access error as $\epsilon=\frac{|a^a_{A,R}-\hat{a^a_{A,R}}|}{a^a_{A,R}}$, where $\hat{a}^a_{A,R}$ is the estimated spectrum access probability. Fig.~\ref{cognitive} shows the results of energy efficiency considering spectrum access probability and spectrum access error, which considers different spectrum access probabilities for the source-to-relay link. From this figure, it is observed that the energy efficiency performs worse with the spectrum access error than that with ideal spectrum access probability and the gap sizes of energy efficiencies first increase and then decrease as the spectrum access probability $a^a_{A,R}$ grows, which means the throughput increased faster than the power consumption at first.} This is because the idle power consumption increases as the spectrum access probability decreases. It could be concluded from this figure that, the throughput cannot increase as the time allocation ratio $\alpha$ is at large region, although the spectrum access probability increases. In fact, the packet dropping rate of relay-buffer increases as the spectrum access probability increases when the source-to-relay time allocation ratio is at large region, causing the decreasing throughput.
\textcolor{black}{Fig.~\ref{direct_relay} illustrates the influences of average SNR to the energy efficiencies of the direct transmission and relay-assisted transmission for different channel parameter $m$, where we assume that the direct link transmission and relay-assisted transmission have the same average SNR at the destination node. We observe that with the increasing $\bar{\mathcal{S}}$, the energy efficiency of direct transmission gradually lose its superiority in energy efficiency. This corresponds to the analytical results in Theorem~\ref{t2}.}
\begin{figure}[!t]
\begin{center}
\includegraphics [width=2.5in]{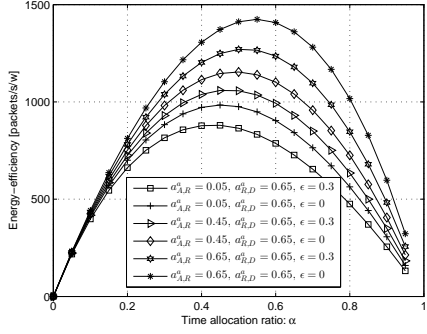}
\caption{The average energy efficiency versus the time allocation ratio $\alpha$ for relay assisted transmission considering the spectrum access probability.} \label{cognitive}
\end{center}
\end{figure}
\begin{figure}[!t]
\begin{center}
\includegraphics [width=2.5in]{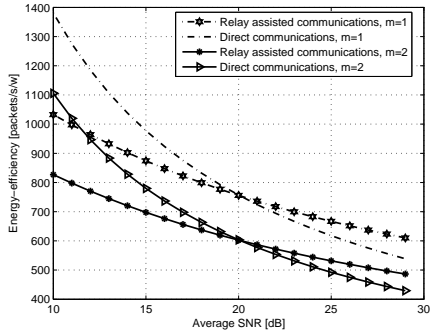}
\caption{The energy efficiencies of direct transmission and relay-assisted transmission, $\alpha=\alpha^{\ast}$.} \label{direct_relay}
\end{center}
\end{figure}

Fig.~\ref{transmit_power} shows the results of the total power consumption in relay-assisted transmission systems including transmit power and idle state power. From this figure we can observe that the total power consumption is increasing with increasing SNR, since the transmit powers for the source node and relay node are increasing with increasing SNR. We can also observe that the power is larger when the time allocation ratio is larger, which comes from the fact that the holding probability for the source to the relay transmission is larger than that of the relay to the destination transmission. At the same time, the gap of the power consumption is increasing when increasing SNR for different time allocation ratios, since larger $\alpha$ has smaller packet dropping rate for the relay-assisted transmission system at large SNR. Then the transmit rate is larger. Based on~\eqref{power}, we can get the result that the total power consumption is larger when $\alpha$ is larger.
\begin{figure}[!t]
\begin{center}
\includegraphics [width=2.5in]{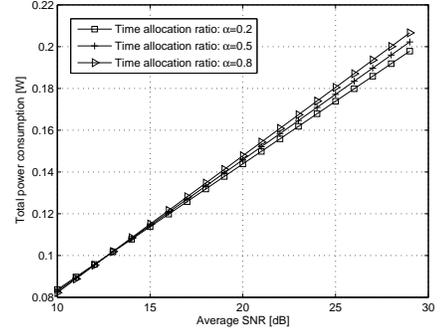}
\caption{The results of the average total power consumption versus average SNR for relay assisted transmission under Nakagami-$m$ fading channel with $m=1$.} \label{transmit_power}
\end{center}
\end{figure}

In Fig.~\ref{Pd} and Fig.~\ref{Pd_buffer}, we show the packet dropping rate of source node $A$, $P^A_d$, for the proposed relay assisted transmission. We can see in Fig.~\ref{Pd} that the packet dropping rate is decreasing when increasing the SNR, since the service rate for the source node is increasing when increasing the SNR. Moreover, we can observe that when the ratio of the source-to-relay time allocation over the relay-to-destination time allocation becomes larger, the packet dropping rate for source node $A$ becomes smaller. When the source-to-relay time allocation is increasing, the dropped number of packets for the source node during each transmission period is decreasing. Then the packet dropping rate becomes smaller. \textcolor{black}{Similarly, it can be seen from Fig.~\ref{Pd_buffer} that the packet dropping rate is decreasing when increasing the buffer size. Due to smaller dropped number of packets, the larger source-to-relay time allocation ratio leads to smaller packet dropping rate.}
\begin{figure}[!t]
\begin{center}
\includegraphics [width=2.5in]{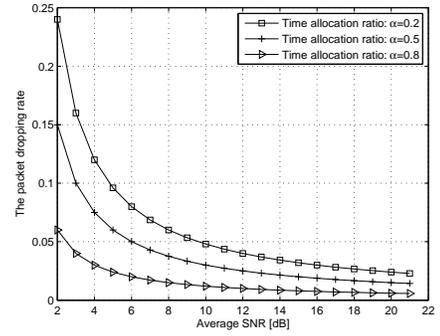}
\caption{The results of the packet dropping rate versus average SNR for relay assisted transmission with different time allocation ratio.} \label{Pd}
\end{center}
\end{figure}
\begin{figure}[!t]
\begin{center}
\includegraphics [width=2.5in]{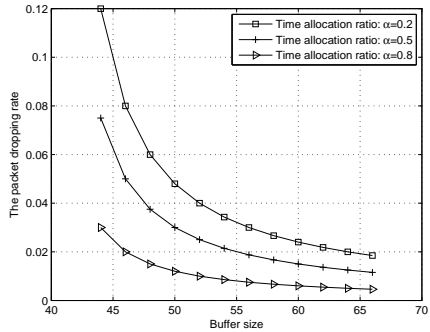}
\caption{The results of the packet dropping rate versus buffer size.} \label{Pd_buffer}
\end{center}
\end{figure}
In Fig.~\ref{delay}, we show the delay performance for the relay-assisted transmission system when the average packet arrival rate $\bar{\lambda}=2$. From this figure we can see that when the SNR is increasing, the delay is decreasing, which matches the results of Lemma~\ref{delay-snr}. We can also see that the delay is smaller when the time allocation ratio is larger. This phenomenon can be understood as follows: when the packet arrival rate is larger, the packet dropping rate of the source node is in dominant place with small time allocation ratio. Then the delay can be smaller when increasing the source-to-relay time allocation.
\begin{figure}[!t]
\begin{center}
\includegraphics [width=2.5in]{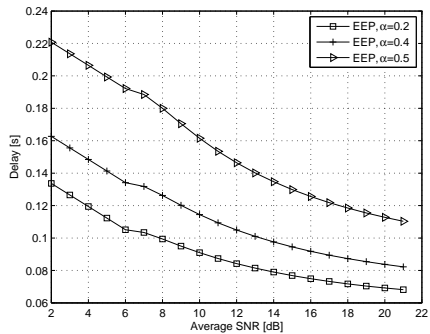}
\caption{The results of the delay for relay assisted transmission with different time allocation ratio.} \label{delay}
\end{center}
\end{figure}

\textcolor{black}{In Fig.~\ref{ee_snr}, the impact of the arrival traffic on the energy efficiency is studied for the relay-assisted transmission, when the time allocation ratio $\alpha=0.6$. This figure shows the results of analysis and simulated values of energy efficiency for different average traffic rates. It can be observed that the variations of analytical and simulation results agree reasonably well.} This figure shows that the energy efficiency is decreasing when increasing the SNR for given average arrival traffic $\bar{\lambda}$ in the regime of large SNR. On the other hand, the energy efficiency is increasing when increasing SNR in the regime of small SNR. In this case, we can get the energy efficient SNR by Algorithm \ref{alg:framwork}. Thus, the energy efficient transmission policy $\omega^{\mathrm{opt}}$ can be obtained, which also guarantees the delay requirement based on \eqref{delay-relay} and \eqref{adelay}. For the case of different $\bar{\lambda}$, the energy efficiency is increasing when increasing the average arrival traffic rate $\bar{\lambda}$. This comes from the fact that the throughput is larger for larger arrival traffic rate.
\begin{figure}[!t]
\begin{center}
\includegraphics [width=2.5in]{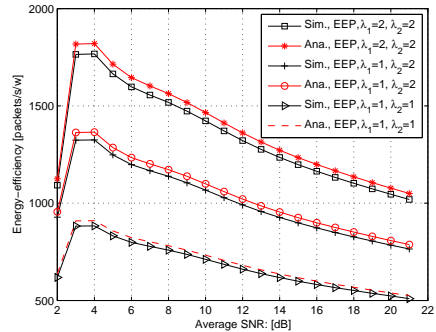}
\caption{The performance of energy efficiency versus average SNR for relay assisted transmission.} \label{ee_snr}
\end{center}
\end{figure}

\section{Conclusion }
In this paper, we present a cross-layer framework in the cognitive cooperative network that determines the energy efficient transmission policy based on both the physical and the upper-layer information. We derive the closed-form expression of delay for the relay-assisted transmission system, which considers the queueing delay of the relay-buffer and source-buffer. With regard to the physical layer transmission, we propose a switching strategy between the direct transmission and the relay-assisted transmission in the cooperative transmission system under different delay requirements. Furthermore, for the relay-assisted transmission, we obtain the energy efficient transmission time allocation ratio for the inter-node link. At last, we get the transmission policies for the source node and relay node that achieve the energy efficient cooperative transmissions for different delay-aware services. Future work is in progress to consider more general case that arbitrary relay nodes are deployed to assist in conveying the source packet in the proposed framework.

\appendices
\section{proof of Lemma~\ref{total-service}}\label{a1}
According to the packet error rate $P_1$ in the first phase, the packet service time $\delta_n$ in state $n$ has the following probability mass function based on the M/G/1 queue model:
\begin{equation}
P\left \{\delta_n=\epsilon_n+\top\bar{\epsilon}_n \right \}=(1-P_1)P_1^{\top}, \top=0,1,\cdots,N^{\max}_r,\label{r-st}
\end{equation}
where $\epsilon_n$, $\bar{\epsilon}_n$ and $\top$ represent the packet transmission time when the channel is in state $n$, the average packet transmission time for retransmissions and the retransmission times, respectively.
When the packet is transmitted from the source-to-destination link, the packet successfully received time at channel state $n$ and average retransmission time for each retransmission are
\begin{equation}
\begin{aligned}
\tau_n^{l_{A,D}}=\frac{L}{b_nR_{s}a^a_{A,D}},\, \bar{\tau}_n^{l_{A,D}}=\sum_{n=1}^N\tau_n^{l_{A,D}}P^{S,D}_r(n),
\end{aligned}
\end{equation}
respectively. 
Due to the independence of inter-node link, the channel state of relay-to-destination link would stay at any state of $\{1,2,\cdots,N\}$ when source-to-destination link is at state $n$. Then, we use the average rate of $\bar{b}=\sum_{n=1}^{N}b_n P_r(n)$ as the transmission modulation for relay-to-destination link. Similarly, the packet successfully received time from relay-assisted transmission at channel state $n$ and average retransmission time for each retransmission are
\begin{equation}
\begin{aligned}
\tau_n^{l_{A,R},l_{R,D}}=\frac{L}{b_nR_{s}a^a_{A,R}}+\frac{L}{\bar{b}R_{s}a^a_{R,D}}, \\ \text{and}\quad\bar{\tau}_n^{l_{A,R},l_{R,D}}=\sum_{n=1}^N\tau_n^{l_{A,R},l_{R,D}}P_r(n),
\end{aligned}
\end{equation}
respectively.
In the sequel, the average transmission time of the packet at channel state $n$ and average retransmission time for each retransmission are
\begin{equation}
\begin{aligned}
\epsilon_n=\tau_n^{l_{A,D}}(1-P_{LD})+\tau_n^{l_{A,R},l_{R,D}}P_{LD}, \\ \text{and}\quad\bar{\epsilon}_n=\bar{\tau}_n^{l_{A,D}}(1-P_{LD})+\bar{\tau}_n^{l_{A,R},l_{R,D}}P_{LD},
\end{aligned}
\end{equation}
respectively,
which can be denoted as \eqref{r-tau}.

Based on~\eqref{r-st}, the average service time at channel state $n$ for the source transmission can be derived as
\begin{equation}\small
\begin{aligned}
E\left \{ \delta_n \right \}&=\epsilon_n (1-P_1)+\sum_{\top=1}^{N^{\max}_r}(\epsilon_n+\top\bar{\epsilon}_n) (1-P_1)P_1^{\top}\\
&=\frac{f(\epsilon_n)+\bar{\epsilon}_nP_1\left [ 1-P_1^{N^{\max}_r}\left(1+N^{\max}_r (1-P_1)\right) \right]}{1-P_1},\label{r-mst}
\end{aligned}
\end{equation}
where $f(\epsilon_n)=\epsilon_n(1-P^{N^{\max}_r+1}_1)(1-P_1)$.

Recall that the ratio of the time allocation in the phase of source-to-relay over the whole period is $\alpha$. To see this, based on~\eqref{r-tau} and~\eqref{r-mst}, we can derive the service rate $\chi_n$ at channel state $n$ as
\begin{equation}\small
\begin{aligned}\label{r-srn}
\chi_n=\frac{\alpha}{E\left \{ \delta_n \right \}}=\frac{\alpha(1-P_1)}{f(\epsilon_n)+\bar{\epsilon}_nP_1\left [ 1-P_1^{N^{\max}_r}\left(1+N^{\max}_r (1-P_1)\right) \right]}.
\end{aligned}
\end{equation}

\section{proof of Theorem~\ref{t1}}\label{a2}
Since the average transmission rate from the source to the relay is $\bar{\chi}$, the average packet arrival rate at the relay is $\bar{\chi}$, which can be derived as $\bar{\chi}=\sum_{n=1}^{N}\chi_n P_r(n)$ based on Lemma~\ref{total-service}. Then, we can get the probability distribution of the instantaneous arrival rate $a_t^R$ at the relay as
\begin{equation}\label{possion}
P\left(a_t^R=a\right)=\begin{cases}
P_r(n), & \text{ if } a\in\left\{\chi_1,\cdots, \chi_N\right\}; \\
\frac{\lambda^ae^{-\lambda}}{a!}, & \text{ otherwise. }
\end{cases}
\end{equation}
Assume that the queue state of the relay node is $S^R_t$ at the start of the $t$-th time slot, and $S^R_t\in \mathbb{S}=\left\{s^{'}_0=0, s^{'}_1=1, \cdots, s^{'}_M=M\right\}$. While AMC transmission is used by the relay, it is known that $\chi_t^{'}\in \Psi^{'}=\left\{\chi^{'}_1, \cdots, \chi^{'}_N\right\}$ can be denoted as the number of packets removed from the queue of the relay. Then, to obtain the resulting transition of the queue state at the relay, both buffer size and the service rate of relay node need to be considered, shown as
\begin{equation}
S^R_t=\mbox{min}\left\{M, \mbox{max}\left\{0, S^R_{t-1}-\chi_t^{'}\right\}+\bar{\chi}\right\}.
\end{equation}

In order to express the system state transition, let $(\chi_t^{'}, S^R_{t-1})$ denote the joint service rate and the queue state, and let $\ss_{(\chi^{'}_x,s^{'}_q),(\chi^{'}_y,s^{'}_l)}$ denote the transition probability from $(\chi_t^{'}=\chi^{'}_x,S^R_{t-1}=s^{'}_q)$ to $(\chi_{t+1}^{'}=\chi^{'}_y,S^R_t=s^{'}_l)$, where $(\chi^{'}_x,s^{'}_q)\in \Psi^{'}\times \mathbb{S}$, and $(\chi^{'}_y,s^{'}_l)\in \Psi^{'}\times \mathbb{S}$. Then, the state transition probability matrix of the relay transmission system can be organized in a block form as
\begin{equation}
\mathbf{\Upsilon}=\left [ \Theta_{\chi^{'}_i,\chi^{'}_j} \right ],1\leq i,j\leq N,
\end{equation}
where
\begin{equation}\small
\Theta_{\chi^{'}_i,\chi^{'}_j}=\begin{bmatrix}
\Xi_{(\chi^{'}_1),(\chi^{'}_1)} & \Xi_{(\chi^{'}_1),(\chi^{'}_2)} & \cdots & \Xi_{(\chi^{'}_1),(\chi^{'}_N)}\\
\Xi_{(\chi^{'}_2),(\chi^{'}_1)} & \Xi_{(\chi^{'}_2),(\chi^{'}_2)} & \cdots & \Xi_{(\chi^{'}_2),(\chi^{'}_N)}\\
\vdots & \vdots & \ddots  & \vdots\\
\Xi_{(\chi^{'}_N),(\chi^{'}_1)} & \Xi_{(\chi^{'}_N),(\chi^{'}_2)} & \cdots & \Xi_{(\chi^{'}_N),(\chi^{'}_N)}
\end{bmatrix},
\end{equation}
the submatrix $\Xi_{(\chi^{'}_x),(\chi^{'}_y)}$ can be shown in~\eqref{transm-relay}.
\begin{figure*}[!t]
\normalsize
\begin{equation}\small\label{transm-relay}
\Xi_{(\chi^{'}_x),(\chi^{'}_y)}=\begin{bmatrix}
\ss_{(\chi^{'}_x,s^{'}_0),(\chi^{'}_y,s^{'}_0)} & \ss_{(\chi^{'}_x,s^{'}_0),(\chi^{'}_y,s^{'}_1)} & \cdots & \ss_{(\chi^{'}_x,s^{'}_0),(\chi^{'}_y,s^{'}_M)}\\
\ss_{(\chi^{'}_x,s^{'}_1),(\chi^{'}_y,s^{'}_0)} & \ss_{(\chi^{'}_x,s^{'}_1),(\chi^{'}_y,s^{'}_1)} & \cdots & \ss_{(\chi^{'}_x,s^{'}_1),(\chi^{'}_y,s^{'}_M)}\\
\vdots & \vdots & \ddots  & \vdots\\
\ss_{(\chi^{'}_x,s^{'}_M),(\chi^{'}_y,s^{'}_0)} & \ss_{(\chi^{'}_x,s^{'}_M),(\chi^{'}_y,s^{'}_1)} & \cdots & \ss_{(\chi^{'}_x,s^{'}_M),(\chi^{'}_y,s^{'}_M)}
\end{bmatrix}.
\end{equation}
\hrulefill
\vspace*{10pt}
\end{figure*}

Finally, the system state transition probability $\ss_{(\chi^{'}_x,s^{'}_q),(\chi^{'}_y,s^{'}_l)}$ can be obtained as~\eqref{transpr-relay},
\begin{figure*}[!t]
\normalsize
\begin{equation}\small\label{transpr-relay}
\begin{aligned}
\ss_{(\chi^{'}_x,s^{'}_q),(\chi^{'}_y,s^{'}_l)}&=\ss(\chi^{'}_{t+1}=\chi^{'}_y,S^{R}_t=s^{'}_l|\chi^{'}_t=\chi^{'}_x,S^{R}_{t-1}=s^{'}_q)\\
&=\ss(\chi^{'}_{t+1}=\chi^{'}_y|\chi^{'}_t=\chi^{'}_x)\ss(S^{R}_t=s^{'}_l|S^{R}_{t-1}=s^{'}_q,\chi^{'}_t=\chi^{'}_x)\\
&=\ss_{\chi^{'}_x,\chi^{'}_y}\ss(S^{R}_t=s^{'}_l|S^{R}_{t-1}=s^{'}_q,\chi^{'}_t=\chi^{'}_x).
\end{aligned}
\end{equation}
\hrulefill
\vspace*{10pt}
\end{figure*}
where $\ss_{\chi^{'}_x,\chi^{'}_y}$ represents the channel state transition probability. Due to the fact that the channel state transition is independent of other states, the second equality in~\eqref{transpr-relay} can be obtained.

By considering the slow fading channel, the channel state transition only happens between adjacent states. Then, the transition probability from non adjacent states is zero, and the nonzero elements $\ss_{\chi^{'}_x,\chi^{'}_y}$ is described in~\cite{2004-Liu-AMC}.
In addition, with regard to the conditional probability of $\ss(S^{R}_t=s^{'}_l|S^{R}_{t-1}=s^{'}_q,\chi^{'}_t=\chi^{'}_x)$, it can be derived in~\eqref{transpp-relay}. In all, based on~\eqref{transpr-relay} and~\eqref{transpp-relay}, the transition probability of the system state for the relay node can be obtained.
\begin{figure*}[!t]
\normalsize
\begin{equation}\small\label{transpp-relay}
\begin{aligned}
\ss(S^{R}_t=s^{'}_l|S^{R}_{t-1}=s^{'}_q,\chi^{'}_t=\chi^{'}_x)
=\begin{cases}
P(a_t^R=s^{'}_l-\mbox{max}\left \{0,s^{'}_q-\chi^{'}_x  \right \}), & \text{ if } 0\leq s^{'}_l<M, \\
1-\sum_{0\leq s^{'}_l<M}\ss(S^{R}_t=s^{'}_l|S^{R}_{t-1}=s^{'}_q,\chi^{'}_t=\chi^{'}_x), & \text{ if } s^{'}_l=M.
\end{cases}
\end{aligned}
\end{equation}
\hrulefill
\vspace*{10pt}
\end{figure*}

Notably particularly, by using the results from~\cite[Theorem 4.1]{2007-Ross-Probability}, we obtain the fact that the Markov chain $\left \{(S^{R}_t,\chi^{'}_t), t\geq 0  \right \}$ exists stationary distribution. 
From the definition of irreducible Markov chain, we know that the Markov chain only has one class. 
Under this circumstance, while any state $(\chi^{'}_x,s^{'}_q)\in \Psi^{'}\times\mathbb{S}$ and $(\chi^{'}_y,s^{'}_l)\in \Psi^{'}\times\mathbb{S}$, we only need to prove that $(\chi^{'}_x,s^{'}_q)$ can access $(\chi^{'}_y,s^{'}_l)$.

In the following, we will prove that the Markov chain of relay system is irreducible. Based on the packet arrival probability, we can obtain the transition probability
\begin{equation}
\ss_{(\chi^{'}_x,s^{'}_q)|(\chi^{'}_x,s^{'}_l)}=P(a_t^R=s^{'}_l-\mbox{max}\left\{0,s^{'}_q-\chi^{'}_x\right\}).
\end{equation}
Thus, state $(\chi^{'}_x,s^{'}_q)$ can access state $(\chi^{'}_x,s^{'}_l)$, i.e.,
\begin{equation}\label{f2-relay}
(\chi^{'}_x,s^{'}_q)\Rightarrow (\chi^{'}_x,s^{'}_l)
\end{equation}
has nonzero probability.
Note that the channel has finite states. Thus, there always exists transition path from the channel state $x$ to the state $y$ over the neighbour state. In this case, state $\chi^{'}_x$ can go to state $\chi^{'}_y$, and
\begin{equation}
\ss_{(\chi^{'}_x,s^{'}_l)|(\chi^{'}_y,s^{'}_l)}=P(a_t^R=s^{'}_l-\mbox{max}\left\{0,s^{'}_l-\chi^{'}_x\right\}).
\end{equation}
Therefore, state $(\chi^{'}_x,s^{'}_l)$ can access state $(\chi^{'}_y,s^{'}_l)$, i.e.,
\begin{equation}\label{f3-relay}
 (\chi^{'}_x,s^{'}_l)\Rightarrow (\chi^{'}_y,s^{'}_l)
\end{equation}
also has nonzero transition probability. With the results of~\eqref{f2-relay} and~\eqref{f3-relay}, we prove that any state $(\chi^{'}_x,s^{'}_q)$ can access any state $(\chi^{'}_y,s^{'}_l)$ in the Markov chain of relay system. Then, we obtain that the system $\left \{(S^{R}_t,\chi^{'}_t), t\geq 0  \right \}$ is irreducible.

When implementing the conclusion of~\cite{2007-Ross-Probability}, the states from a finite irreducible Markov chain are recurrent. In all, the stationary distribution of the Markov process $\left \{(S^{R}_t,\chi^{'}_t), t\geq 0  \right \}$ exists.
Then, we can get the stationary distribution by solving
\begin{equation}\label{stationary-relay}
\boldsymbol{\pi^{'}}=\boldsymbol{\pi^{'}}\Upsilon, \sum_{s^{'}\in \mathbb{S},\chi^{'}\in \Psi^{'}}\pi^{'}_{(s^{'},\chi^{'})}=1.
\end{equation}
while the number of arrival packets is larger than that of remaining space, i.e., $\bar{\chi}>M-(S^R_{t-1}-\chi^{'}_t)$, the packets are dropped by the relay-buffer.
By using the stationary distribution, the average number of dropped packets can be expressed as
\begin{equation}\small
\begin{aligned}
E\left\{\Delta^R\right\}&=\sum_{s^{'}\in \mathbb{S}, \chi^{'}\in \Psi^{'}}\Delta^R\cdot\ss(\bar{\chi}|S^R_{t-1}=s^{'}, \chi^{'}_t=\chi^{'})\\
&=\sum_{ s^{'}\in \mathbb{S}, \chi^{'}\in \Psi^{'}}\left [(\bar{\chi}-(M-(S^R_{t-1}-\chi^{'}_t)))\times \pi^{'}_{s^{'},\chi^{'}}\right ].
\end{aligned}
\end{equation}
In all, the packet dropping rate at the relay node is $P^R_d=\frac{E\left\{\Delta^R\right\}}{\bar{\chi}}$.

\section{proof of Lemma~\ref{total-power}}\label{a3}
\textcolor{black}{In Phase source node to relay node, source A transmits data to the relay R with power $\mathcal{P}$. Due to the opportunistic spectrum access, the average transmit power is $a^a_{A,R}\mathcal{P}$, the average circuit power consumption is $a^u_{A,R}\mathcal{P}_0$.}

\textcolor{black}{On the other hand, in Phase relay node to destination node, relay R transmits data to the destination D with power $\mathcal{P}_R$. Similarly, the average transmit power and circuit power of the relay node are $a^a_{R,D}\mathcal{P}_R$ and $a^u_{R,D}\mathcal{P}_0$, respectively. In all, the total energy consumption of the secondary relay-assisted transmission networks is given by
\begin{equation}
\zeta=(a^a_{A,R}\mathcal{P}+a^u_{A,R}\mathcal{P}_0)T_1+(a^a_{R,D}\mathcal{P}_R+a^u_{R,D}\mathcal{P}_0)T_2.
\end{equation}}

\section{proof of Theorem~\ref{t2}}\label{a4}
For the first case, suppose $\mathcal{D}_0\rightarrow 0$, based on lemma~\ref{delay-snr}, the average SNR should satisfy $\bar{\mathcal{S}},\bar{\mathcal{S}}_{R}\rightarrow \infty$, causing $P_d\rightarrow 0$ and $P_{0}\rightarrow 0$. Thus, the average throughput $\mathcal{R}=\bar{\lambda}$. At the same time, the total power consumption is dominated by the transmit power. Consequently, based on Lemma~\ref{total-power}, the total average energy consumption of source-to-destination direct transmission is $\zeta=a^a_{A,D}T\mathcal{P}_{S,D}$. In all, the energy efficiencies of the direct transmission and the relay-assisted transmission are
\begin{equation}\small
\begin{aligned}
\lim_{\mathcal{D}_0\rightarrow 0}\eta_{\mathrm{ee,direct}}&=\frac{\bar{\lambda}}{a^a_{A,D}\mathcal{P}},~\text{and}\\
\lim_{\mathcal{D}_0\rightarrow 0}\eta_{\mathrm{ee,relay}}&=\frac{T\bar{\lambda}}{a^a_{A,R}T_1\mathcal{P}+a^a_{R,D}T_2\mathcal{P}}\\
&=\frac{\bar{\lambda}}{a^a_{A,R}\cdot\alpha\mathcal{P}+a^a_{R,D}(1-\alpha)\mathcal{P}},
\end{aligned}
\end{equation}
respectively. From these two equations, we can obtain that if $a^a_{A,D}>\frac{a^a_{A,R}+a^a_{R,D}\cdot\alpha}{1+\alpha}$, the energy efficiency of relay-assisted transmission is larger that that of the direct transmission. The packet should be transmitted with the cooperative links $l_{A,R}$ and $l_{R,D}$ to ensure energy efficient transmission.

For the second case, as $\bar{\mathcal{S}},\bar{\mathcal{S}}_{R}\rightarrow 0$, suppose $\mathcal{D}_0\rightarrow \infty$, causing $P_{LD}\rightarrow 1$. Then, the idle state power dominates the total power consumption. Then the energy consumptions of direct transmission and relay-assisted transmission are
\begin{equation}
\begin{aligned}
\zeta_{\mathrm{direct}}&=(a^u_{A,D}\mathcal{P}_0)T,~\text{and} \\
\zeta_{\mathrm{relay}}&=(a^u_{A,R}\mathcal{P}_0)T_1+(a^u_{R,D}\mathcal{P}_0)T_2,
\end{aligned}
\end{equation}
respectively. Note that spectrum access probabilities satisfy $a^u_{A,D}+a^a_{A,D}=1$, $a^u_{A,R}+a^a_{A,R}=1$ and $a^u_{R,D}+a^a_{R,D}=1$. Similarly, we can also obtain that if $a^a_{A,D}>\frac{a^a_{A,R}+a^a_{R,D}\cdot\alpha}{1+\alpha}$, the energy efficiency of the direct transmission is larger than that of the relay-assisted transmission.
Therefore, the packet should be transmitted with the cooperative links $l_{A,R}$ and $l_{R,D}$.


\bibliography{mybib}
\bibliographystyle{ieeetr}

\end{document}